\documentclass[runningheads]{llncs}
\usepackage{enumitem}
\usepackage{tikz}
\usepackage{amsmath}
\usepackage{amssymb}
\usepackage{mathtools}
\usepackage{hyperref}
\usepackage{cleveref}

\DeclareMathOperator{\xbetween}{between}
\def\indiff{\simeq}

\def\acceptable{{\mathcal A}}

\newcommand{\GStar}[2]{\Gamma^*(#1, #2)}
\newcommand{\GStarRestrict}[3]{\restrict{\Gamma^*(#1, #2)}{#3}}
\def\GStarMu{\GStar{\mu}{\mu'}}
\def\GStarR{\GStarRestrict{\mu}{\mu'}{i}}
\def\GStarDotMu{\GStar{\mu^j}{\mu^{j+1}}}

\def\GStarRjj{\GStarRestrict{\mu}{\mu'}{j+1}}

\newcommand{\tbdefined}[1]{\textbf{#1}} 
\newcommand{\setdiff}{\backslash}
\newcommand{\setb}{\ : \ } 
\newcommand{\restrict}[2]{\left. #1 \right|_{#2}} 

\newcommand{\abs}[1]{\left\lvert #1 \right\rvert}
\newcommand{\curly}[1]{\left\{ #1 \right\}}
\newcommand{\round}[1]{\left( #1 \right)}
\newcommand{\floor}[1]{\left\lfloor #1 \right\rfloor}

\newcommand\RR{\mathbb R}

\begin{document}

\title{On the diameter of the polytope of the stable
marriage with ties}
%
%
\author{Felix Bauckholt \and Laura Sanit\`{a}}
\authorrunning{F. Bauckholt and L. Sanit\`{a}}
%
\institute{Department of Combinatorics and Optimization, University of Waterloo, \\
Waterloo, ON N2L 3G1, Canada \\
\email{felixbauckholt@gmail.com, lsanita@uwaterloo.ca}}
\maketitle              
\begin{abstract}
The stable marriage problem with ties is a well-studied and interesting problem in game theory.
We are given a set of men and a set of women. Each individual has a preference
ordering on the opposite group, which can possibly contain ties. 
A stable marriage is given by a matching between men and women
for which there is no blocking pair, i.e., a men and a women who strictly 
prefer each other to their current partner in the matching. 

In this paper, we study the diameter of the polytope given by the convex hull of characteristic
vectors of stable marriages, in the setting with ties. 
We prove an upper bound of $\lfloor \frac{n}{3}\rfloor$ on the diameter, where $n$ is the total number 
of men and women, and give a family of instances
for which the bound holds tight. Our result generalizes the bound on the diameter 
of the standard stable marriage polytope (i.e., the well-known polytope that describes the setting without ties), 
developed previously in the literature.

\keywords{Stable matchings \and Diameter of polytopes.}
\end{abstract}
\section{Introduction}
The stable marriage problem is one of the most popular and fundamental problems in game theory.
An instance of the problem can be described by a (non necessarily complete) bipartite graph, 
where the bipartition is on a set $M$ of men and a set $W$ of women.
For each man $m$ (resp. woman $w$) there is a strict ordering defined on the neighboring women (resp. men).
A \emph{stable marriage} is given by a matching between men and women that does not have a \emph{blocking pair}, i.e., 
a pair of individuals that mutually prefer each other to their current partner in the matching.
The problem was introduced in the seminal work of Gale and Shapley~\cite{GaleShapley62}, which shows that a stable marriage always exists, and gives an elegant algorithm to efficiently find one. Since then, the stable marriage problem and its variants have been intensively studied by plenty of researchers in multi-disciplinary contexts, often bridging the areas of applied mathematics, computer science, and economics.

As it is natural to expect, the problem has been widely investigated also from a polyhedral point of view. 
In particular, a Linear Programming (LP)-description of the so-called \emph{stable marriage polytope} is well-known (see \cite{LP-bliss,onthestable,Rothblum,Jochen}). The stable marriage polytope is defined as the convex hull of 
the characteristic vectors of all stable marriages associated to a given instance.
Many structural properties of this polytope have been established in the literature, including a bound on its \emph{diameter}. 
We recall that the diameter is given by the maximum length of a shortest path between two vertices on the 1-skeleton of a polytope\footnote{The 1-skeleton of a polytope $P$ is the graph in which the vertices correspond to the extreme points of $P$, and the edges correspond to the 1-dimensional faces of $P$.}. The diameter is arguably one of the most important polyhedral concepts, and it constitutes a central research topic in discrete mathematics. In particular, bounding the diameter of polytopes that correspond to the set of feasible solutions of fundamental combinatorial optimization problems has been a classical subject of investigation for more than 50 years (just to mention a few, see e.g. diameter results for matchings, TSP, or network flows and transportation in~\cite{BalRus,CHVATAL,bmatch,Padberg1974,RS98,Borgwardt2018,Brightwell2006,Bal,Borgwardt2016,Sanita18}).   
For the stable marriage polytope, which we call $\mathcal{P}_{SM}$ here, Eirinakis et al.~\cite{stable_diameter} proved a diameter upper bound of  $\lfloor{n/4}\rfloor$, 
 where $n:=|M \cup W|$.
The authors also show the existence of instances for which this bound holds tight.

 \smallskip
In this paper, we focus on the \emph{stable marriage problem with ties}, that is an important and highly popular generalization of the stable marriage problem. 
Here for each man $m$ (resp. $w$) there is an ordering
defined on the neighboring women (resp. men), but
unlike the standard stable marriage setting, these orderings are now \emph{weak} linear orders, i.e., they can contain ties.
This generalization allows for more flexibility in modeling practical settings, where the 
assumption of having a strict order is too restrictive (see~\cite{Irving}). However, the presence of ties makes the problem definitely more difficult. In this case, different stable matchings can have different cardinalities, in contrast with the standard setting, and finding one stable matching of maximum cardinality becomes NP-hard (in fact, APX-hard) 
\cite{incomplete_ties,HardVariants,Iwamaetal07}. Consequently, optimizing over the corresponding polytope is hard, and no LP-description for it is known. 

\medskip
\noindent
{\bf Our results and techniques.} In this paper, we give an upper bound of $\lfloor \frac{n}{3} \rfloor$ on the diameter of the polytope of the stable marriage problem with ties (which we call $\mathcal{P}_{SMT}$), and give a family of instances for which our bound holds tight. Our result generalizes what is known for $\mathcal{P}_{SM}$, meaning that if all preference orderings are strict, then it recovers the bound given in~\cite{stable_diameter}. However, it relies on different and new ingredients, which we are going to describe next.

A key tool used in~\cite{stable_diameter} to bound the diameter of $\mathcal{P}_{SM}$ is the so-called \emph{stable marriage graph}, introduced in \cite{Maffray}. The stable marriage graph is an auxiliary graph that one can construct (in polynomial time) for a given instance of the standard stable marriage problem. The author of \cite{onthestable} showed that extreme point adjacency on $\mathcal{P}_{SM}$ can be inferred by looking at the number of nontrivial connected components of a (suitably defined) subgraph of the stable marriage graph. The authors of \cite{stable_diameter} prove that the number of nontrivial components of this subgraph indeed yields an upper bound on the distance between two extreme points of $\mathcal{P}_{SM}$. The proofs of all these results rely crucially on the fact that the stable marriages (in the setting without ties) form a so-called \emph{distributive lattice} (see \cite{lattice}). Such
arguments do not apply in our setting since stable marriages in the presence of ties do not have this nice property. In addition, we do not have an explicit LP-description of the polytope.

Despite this, we show that one can still give a graphical characterization of adjacency for the extreme points
of $\mathcal P_{SMT}$. To this end, we extend the definition of the stable marriage graph and its relevant subgraphs to the setting with ties,
and obtain properties similar to the ones used in \cite{stable_diameter}. Namely,
(i) two extreme points are adjacent on $\mathcal{P}_{SMT}$ if and only if a suitably defined subgraph of our stable marriage graph
has only one nontrivial connected component, and (ii) one can construct a path between these points by switching coordinates on
one connected component of this subgraph, at each step. Proving condition (i), in particular, requires new arguments. 
Given two extreme points $x$ and $x'$ of $\mathcal P_{SMT}$, we identify some inequalities that are valid for our polytope and are tight for both $x$ and $x'$, exploiting the standard LP-relaxation for $\mathcal{P}_{SMT}$. To infer adjacency, we look at the \emph{span} of the corresponding coefficient vectors, and use special subgraphs of the marriage graph (which we call \emph{principal blocks}) to find a subset of linearly independent vectors of sufficiently large cardinality. 

As a corollary of our arguments, it follows that although $\mathcal{P}_{SMT}$ models an NP-hard optimization problem, testing adjacency of two extreme points can be done in polynomial time.
In fact, the existence of an efficient characterization of vertex adjacency for polytopes that model NP-hard combinatorial problems
is a subject of research of independent interest (see e.g. the classical hardness result for TSP~\cite{Papadimitriou1978}, and more results in~\cite{CHVATAL,matsui,MATSUIA1995311,Hausmann1978,AGUILERA201740}). As mentioned in some of those papers,
results in this area have the potential to be exploited algorithmically, e.g. in the context of local search techniques.

In a nutshell, with this paper we add the polytope of the stable marriage problem with ties to (i) the list of polytopes
modeling combinatorial problems for which non-trivial bounds on the diameter have been given, and (ii)
the list of polytopes modeling NP-hard problems, for which testing extreme point adjacency can be done in polynomial time. 
In addition, we hope that the structural and graphical insights on the problem developed here could be of interest 
beyond our particular result, and e.g. be useful 
also from an algorithmic perspective.

\section{Preliminaries and Notation}
\label{sec:preliminaries}
We are going to represent an instance of the stable marriage problem with ties as follows. For each man $m$ (resp. woman $w$), we assume to have an ordering $\leq_m$ (resp. $\leq_w$) defined on some subset $P(m) \subseteq W$ of women (resp. $P(w) \subseteq M$ of men). Note that
the subsets can be \emph{strict}, i.e., we are dealing with the general case of (possibly) incomplete lists. 
We write $w \leq_m w'$ if $m$ weakly prefers $w'$ to $w$, $w <_m w'$ if $m$ strictly prefers $w'$
to $w$, 
and we write $w \indiff_m w'$ if $m$ is indifferent between $w$ and $w'$ (that is, $w \leq_m w'$ and $w'
\leq_m w$).
A pair $(m, w)$ is called an \emph{acceptable pair} if $m \in P(w)$ and $w
\in P(m)$. We let $\acceptable \subseteq M \times W$ be the set of all acceptable
pairs.

A \emph{matching} $\mu \subseteq \acceptable$ is simply a set of acceptable man-woman pairs
such that each man and each woman occurs in at most one pair. For convenience, we will also regard $\mu$
as a function, so that $\mu(m) = w$ and $\mu(w) = m$ for every $(m, w) \in \mu$.
Whenever a man or woman occurs in no pairs of $\mu$, that man or woman is said
to be \emph{single}. For each single man $m$ (resp. woman $w$), we
define $\mu(m) = \bot$ (resp. $\mu(w) = \bot$), and assume $\bot <_m w$ for each $w \in P(m)$
(resp. $\bot <_w m$ for each $m \in P(w)$).
There are several ways of generalizing the concept of a stable matching to the
scenario with ties. We follow the common literature in using
weak stability, as defined by \cite{Irving}.

\begin{definition}
A matching $\mu$ is \emph{stable}, if for every acceptable pair $(m, w) \in
\acceptable$, we have $m \leq_w \mu(w)$ or $w \leq_m \mu(m)$.
\end{definition}
A pair $(m, w) \in \acceptable$ for which this condition fails, that is, where
$m >_w \mu(w)$ and $w >_w \mu(m)$, is a blocking pair.

The polytope $\mathcal{P}_{SMT}$ is defined to be the convex hull
of all points $x_\mu \in \RR^\acceptable$ 
for all stable matchings $\mu$, where
we define $x_\mu$ such that $x_\mu((m, w)) = 1$ if $(m, w) \in \mu$ and
$x_\mu((m, w)) = 0$ otherwise.

In the standard stable marriage setting, $\mathcal{P}_{SM}$ can be described by a
list of inequalities whose size is linear in $|\acceptable|$  \cite{LP-bliss,Rothblum}. In the presence of ties, generalizations of those inequalities do not provide an exact description of our polytope, but they still provide a valid LP-relaxation for  $\mathcal{P}_{SMT}$, that we are going to use in this paper. Here is the set of such inequalities:

\begin{align}
\label{stuffgeq0}
x((m, w)) &\geq 0 && \text{for all } (m, w) \in \acceptable \\
\label{mleq1}
\sum_{\mathclap{w \in W: (m, w) \in \acceptable}} x((m, w)) &\leq 1 && \text{for all } m \in M \\
\label{wleq1}
\sum_{\mathclap{m \in M: (m, w) \in \acceptable}} x((m, w)) &\leq 1 && \text{for all } w \in W \\
\label{theweirdinequality}
x((m, w))
+ \sum_{\mathclap{\substack{m' \in M \setdiff \curly m: \\ (m', w) \in \acceptable, \\ m' \geq_w m}}} x((m', w))
+ \sum_{\mathclap{\substack{w' \in W \setdiff \curly w: \\ (m, w') \in \acceptable, \\ w' \geq_m w}}} x((m, w'))
&\geq 1 && \text{for all } (m, w) \in \acceptable
\end{align}

Intuitively, for a $0/1$ vector $x$, the first three sets of inequalities enforce that $x$ is a matching, while the fourth one enforces
that every pair $(m,w)$ cannot be a blocking pair (either $(m,w)$ is in the matching, or at least one between $w$ and $m$ 
has to be matched with an equally or better ranked partner).

For any given (possibly directed) graph $H$, we let $V(H)$ denote its vertex set and $E(H)$ denote its edge set.
A vertex $v \in V(H)$ is called a successor of $u \in V(H)$ if there is a directed edge in $E(H)$ whose head is $v$
and whose tail is $u$. A component of $H$ is a maximal (weakly) connected subgraph of $H$. 
A component of a graph is called trivial if it is a singleton vertex, and nontrivial otherwise.

In the following, we will use the symbol $\land$ to denote the AND logical operator, and the symbol $\lor$ to denote the OR logical operator.

\section{The marriage graph}

The marriage graph defined in~\cite{Maffray} naturally extends to the setting with ties, as follows.

\begin{definition}
The \tbdefined{marriage graph} $\Gamma$ is a directed graph with vertex set
$V(\Gamma) := \acceptable$, and edge set
\begin{align*}
E(\Gamma) :=&
\curly{((m, w), (m, w')) \setb \Big( (m, w), (m, w') \in V(\Gamma) \Big)  \land \Big( w \neq w' \Big)  \land \Big( w \leq_m w' \Big) }
\\ \cup&
\curly{((m, w), (m', w)) \setb \Big( (m, w), (m', w) \in V(\Gamma) \Big) \land \Big( m \neq m' \Big) \land \Big( m \leq_w m' \Big)}.
\end{align*}
\end{definition}

We say that an edge of the form $((m, w), (m, w'))$ \emph{represents the
preference of the man $m$}, while an edge of the form $((m, w), (m', w))$
represents the preference of the woman $w$. See Fig.~\ref{pic:marriage_graph} for an example.

\begin{figure}
\caption{Example of a marriage graph $\Gamma$.\hspace{8cm}}
\label{pic:marriage_graph}
\begin{center}
\begin{tikzpicture}[scale=1.8]
\tikzstyle{every node}=[rectangle,draw=black]
\tikzstyle{edge}=[->, >=latex, line width=1pt, blue]
\node (At) at (0, 0) {$m_3,w_1$};
\node (Ab) at (1, 0) {$m_3, w_2$};
\node (Bt) at (2, 1) {$m_2,w_3$};
\node (Bb) at (3, 1) {$m_2, w_4$};
\node (As) at (1, 2) {$m_1, w_2$};
\node (Bs) at (3, 2) {$m_1,w_4$};
\node (L) at (0, 1) {$m_2,w_1$};
\draw[edge, red, ->] ([yshift=0.4mm]At.east) -- ([yshift=0.4mm]Ab.west);
\draw[edge, red, <-] ([yshift=-0.4mm]At.east) -- ([yshift=-0.4mm]Ab.west);
\draw[edge, red, ->] ([xshift=-0.4mm]As.south) -- ([xshift=-0.4mm]Ab.north);
\draw[edge, red, <-] ([xshift=0.4mm]As.south) -- ([xshift=0.4mm]Ab.north);
\draw[edge, red, ->] ([yshift=0.4mm]Bt.east) -- ([yshift=0.4mm]Bb.west);
\draw[edge, red, <-] ([yshift=-0.4mm]Bt.east) -- ([yshift=-0.4mm]Bb.west);
\draw[edge, red] (Bs) -- (Bb);
\draw[edge, red] (Bs) -- (As);
\draw[edge, red] (Bt) -- (L);
\draw[edge, red] (At) -- (L);
\draw[edge, red, bend right] (Bb.north west) to [out=-30, in=-150] (L.north east);
\end{tikzpicture}
\end{center}
The above marriage graph represents the instance with $M=\{m_1, m_2, m_3\}$, $W=\{w_1, w_2, w_3, w_4\}$, and the following 
preferences: \\
$P(m_1) = \{w_2,w_4\}$ with $w_4 <_{m_1} w_2$;\\
$P(m_2) = \{w_1,w_3,w_4\}$ with $w_3 <_{m_2} w_1$, $w_4 <_{m_2} w_1$, $w_3 \indiff_{m_2} w_4$;\\
$P(m_3) = \{w_1,w_2\}$ with $w_1 \indiff_{m_3} w_2$;\\
$P(w_1) = \{m_2,m_3\}$ with $m_3 <_{w_1} m_2$;\\
$P(w_2) = \{m_1,m_3\}$ with $m_1 \indiff_{w_2} m_3$;\\
$P(w_3) = \{m_2\}$;\\
$P(w_4) = \{m_1,m_2\}$ with $m_1 <_{w_4} m_2$.\\
\end{figure}
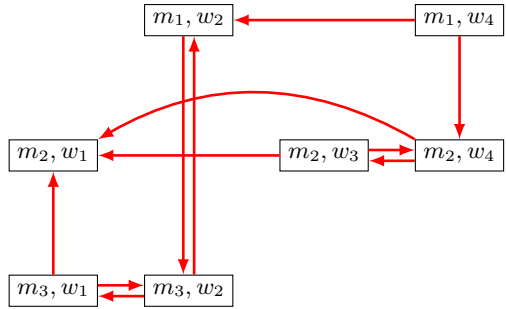

Similarly to~\cite{onthestable}, we would like to characterize adjacency on $\mathcal P_{SMT}$ by looking at some suitable subgraph
of $\Gamma$. To this end, we properly tweak the subgraph definition that is used in the standard stable marriage setting.
Specifically, given two stable matchings $\mu$ and $\mu'$, 
we define a subgraph generated by
the matchings as follows:

\begin{definition}
For two stable matchings $\mu$ and $\mu'$, we define $\GStarMu$ to be the
subgraph of $\Gamma$ induced by the vertex set
\[
\curly{(m, w) \in V(\Gamma) \setb \Big( \xbetween_m(w, \mu(m), \mu'(m)) \Big) \land \Big(\xbetween_w(m, \mu(w),
\mu'(w)\Big)}
\]
where $\xbetween_d(a, b, c)$ is the following boolean expression
\[
\xbetween_d(a, b, c) :=
\Big(a = b \Big) \lor \Big(a = c \Big) \lor \Big( b <_d a \leq_d c \Big) \lor \Big( c <_d a \leq_d b \Big).
\]
 \end{definition}

If there are no ties, this definition is equivalent to the more elegant
definition in \cite{onthestable,stable_diameter}. In the presence of ties, our definition
allows us to recover the following important property.

\begin{lemma}
\label{trivial_component_GStarMu}
For any $v \in V(\Gamma)$, $v$ is an isolated vertex of $\GStarMu$ if and only if $v \in \mu \cap \mu'$.
\end{lemma}

\begin{proof}
For the first direction, let $v$ be an isolated vertex of $\GStarMu$.
It is enough to prove that $v \in \mu$. Suppose not. Since $v$ is not a blocking
pair, we must have some vertex $v' \in \mu$ ($v' \neq v$) such that $(v, v') \in E(\Gamma)$.
Without loss of generality, suppose that the edge $(v, v')$ represents the preference of a man $m$.
Let $v = (m, w)$, $v' = (m, w')$. Note that $\mu(m) = w'$.
Note that we have $\xbetween_m(w', \mu(m), \mu'(m))$ since $w' = \mu(m)$. We also have
$\xbetween_{w'}(m, \mu(w'), \mu'(w'))$ since $m = \mu(w')$. So $v' \in
V(\GStarMu)$. Since $v \in V(\GStarMu)$, we also have $(v, v') \in E(\GStarMu)$, contradicting that
$v$ is isolated.

For the second direction, assume that $v \in \mu \cap \mu'$. Clearly, we
have $v \in V(\GStarMu)$. Suppose for a contradiction that $v$ is not an
isolated vertex. Let $v'$ be a neighbor of $v$ in $\GStarMu$, and without loss of generality, 
assume that the edge between $v$ and $v'$ represents the preference of a
man $m$. Let $v = (m, w)$ and $v' = (m, w')$, where $w \neq w'$. Note that $w =
\mu(m) = \mu'(m)$. Then, either $w' \leq_m w$ or $w' >_m w$. In both cases,
we do not have $\xbetween_m(w',\mu(m), \mu'(m))$. So $v' \notin V(\GStarMu)$, again a contradiction.
\qed
\end{proof}

\smallskip
We now introduce the following further definitions.
For any man $m$, we let $P_{\mu, \mu'}^m$ be the subgraph of
$\GStarMu$ induced by all vertices involving $m$, i.e. induced by $\curly{(m', w') \in V(\GStarMu) \setb m' = m}$.
Similarly, for any woman $w$, we let $P_{\mu, \mu'}^w$ be the subgraph of
$\GStarMu$ induced by $\curly{(m', w') \in V(\GStarMu) \setb w' = w}$.
See Fig.~\ref{pic:principal_block} for an example.

\begin{figure}
\caption{Example of $\GStarMu$ and a principal block.\hspace{8cm}}
\label{pic:principal_block}
\begin{center}
\begin{tikzpicture}[scale=1.5]
\tikzstyle{every node}=[rectangle,draw=black]
\tikzstyle{edge}=[->, >=latex, line width=1pt, blue]
\node (Ab) at (1, 0) {$m_3, w_2$};
\node (As) at (1, 2) {$m_1, w_2$};
\node (Bs) at (3, 2) {$m_1,w_4$};
\node (L) at (0, 1) {$m_2,w_1$};
\draw[edge, red, ->] ([xshift=-0.4mm]As.south) -- ([xshift=-0.4mm]Ab.north);
\draw[edge, red, <-] ([xshift=0.4mm]As.south) -- ([xshift=0.4mm]Ab.north);
\draw[edge, red] (Bs) -- (As);
\end{tikzpicture}
\end{center}
The above graph represents $\GStarMu$ for $\mu=\{(m_2, w_1),(m_1, w_2)\}$ and $\mu'=\{(m_2, w_1),(m_3, w_2),(m_1, w_4)\}$, with respect to the instance described in Fig.~\ref{pic:marriage_graph}.  
The pair of vertical edges yields the edges of the principal block $P^{w_2}_{\mu,\mu'}$.
\end{figure}
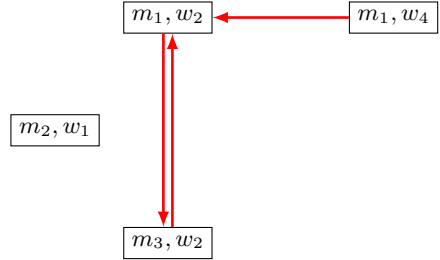

\begin{definition}
We call an induced subgraph $P$ of $\GStarMu$ a \emph{principal
block of $\GStarMu$} if $P = P_{\mu, \mu'}^m$ for some man $m$ or
$P = P_{\mu, \mu'}^w$ for some woman w.
\end{definition}

We have the following trivial observation.
\begin{proposition}
\label{edge_to_circuit}
For every edge $e \in E(\GStarMu)$, there is a principal block $P$ such that $e
\in P$.
\end{proposition}

\section{A characterization of adjacency}

The goal of this section is to prove the following theorem.

\begin{theorem}
\label{thetheorem}
Two matchings $\mu$ and
$\mu'$ correspond to adjacent extreme points in $\mathcal P_{SMT}$ if and only if
$\Gamma^*(\mu, \mu')$ has exactly one nontrivial component. 
\end{theorem}

 Note that, since $\Gamma$ (and hence $\Gamma^*(\mu, \mu')$) can be constructed in polynomial time, 
the above theorem implies that adjacency of two extreme points of $\mathcal P_{SMT}$ can be checked in polynomial time.

\subsection{Proof of the ``if'' part.} Let $\mu$ and $\mu'$ be two stable matchings such that $\GStarMu$
has only one nontrivial component. The strategy that we will use to show adjacency is as follows.

Let $E'$ be the set of \emph{all} possible inequalities that are valid for $\mathcal P_{SMT}$ and are tight
for both the extreme points corresponding to $\mu$ and $\mu'$. 
We will represent an element of $E'$ as a vector of the form $(\alpha, \beta)$, where $\alpha \in
\RR^{\acceptable}$ and $\beta \in \RR$,
expressing the inequality $\alpha^T x \geq \beta$.
To prove that $\mu$ and $\mu'$ correspond to two adjacent extreme points, we would like to show
that the set $\bar E:=\{\alpha \in \RR^{\acceptable}: (\alpha, \beta) \in E' \mbox{ for some } \beta \in \RR \}$ contains 
$\abs{\acceptable} - 1$ linearly independent vectors. To this end, let $E$ be the \emph{span} of $\bar E$.
We will prove that $E$ contains $\abs{\acceptable} - 1$ linearly independent vectors.

For an element $v \in \mu$, we let $e_v \in \RR^{\acceptable}$ be the standard basis vector indexed by $v$ (i.e., $e_v$ has value 1 in the entry
corresponding to $v$, and 0 otherwise).
\begin{lemma}\label{lem:obs}
The following holds:
\begin{enumerate}
\item[(a)] For every $v \notin \mu \cup \mu'$, we have $e_v \in E$; 
\item[(b)] For every man $m$ married in both $\mu$ and $\mu'$ with $\mu(m) \neq
\mu'(m)$, we have $(e_{(m, \mu(m))} + e_{(m,
\mu'(m))}) \in E$;
\item[(c)] For every woman $w$ married in both $\mu$ and $\mu'$ with $\mu(w) \neq
\mu'(w)$, we have $(e_{(\mu(w), w)} + e_{(\mu'(w), w)}) \in E$;
\item[(d)] For every pair $(m, w)$ in $V(\GStarMu)$, but not in $\mu \cup
\mu'$, we have either $(e_{(m, \mu(m))} + e_{(\mu'(w), w)}) \in E$, or
$(e_{(m, \mu'(m))} + e_{(\mu(w), w)}) \in E$.
\end{enumerate}
\end{lemma}

\begin{proof}
$(a)$: Note that $(e_v, 0) \in E'$: the validity of the corresponding inequality follows from \eqref{stuffgeq0}. The tightness for both $\mu$ and $\mu'$ is obvious. Therefore, $e_v \in E$.\\
$(b)$: For every man $m$ married in both $\mu$ and $\mu'$ with $\mu(m) \neq
\mu'(m)$, we have $(-e_{(m,\mu(m))} - e_{(m, \mu'(m))}, -1) \in E'$: the validity of the corresponding inequality follows from \eqref{mleq1}. The tightness for both $\mu$ and $\mu'$ is obvious. Thus $(e_{(m, \mu(m))} + e_{(m,\mu'(m))}) \in E$.\\
$(c)$: The argument follows as in $(b)$, relying on \eqref{wleq1} instead of \eqref{mleq1}.\\
$(d)$: To prove $(d)$, we first show the following claim.

\smallskip
\emph{Claim 1}: Either (i) $\mu(m) \geq_m w$ and $\mu'(w) \geq_w m$, or (ii) $\mu'(m) \geq_m w$ and $\mu(w) \geq_w m$.  

\smallskip
\emph{Proof of Claim 1.} 
Since $\mu$ is stable, we either have $\mu(m) \geq_m w$ or $\mu(w) \geq_w m$.
Similarly, since $\mu'$ is stable, we have $\mu'(m) \geq_m w$ or $\mu'(w) \geq_w
m$. Also, we must have $\mu(m) \geq_m w$ or $\mu'(m) \geq_m w$, since $(m, w) \in
V(\GStarMu)$. Similarly, we must have $\mu(w) \geq_w m$ or $\mu'(w) \geq_w m$.
Thus, if one of the two conditions in (i) is false, the two
conditions in (ii) must both be true, and vice versa. \hspace{2cm} $\diamond$

\smallskip
Now consider the following inequality:
\begin{equation}
\label{dumb_ineq}
\round{e_{(m,w)} + \sum_{v \text{ successor of } (m,w) \text{ in }\Gamma} e_v}^T x \geq 1
\end{equation}
It is clearly a valid inequality for all stable
matchings, since it is just a restatement of \eqref{theweirdinequality}.

Assume case (i) of Claim 1 holds. We have $\mu(m) \geq_m w$. Since $(m, w) \notin \mu \cup \mu'$
and $(m, w) \in V(\GStarMu)$, we have $w >_m \mu'(m)$. Thus, only one successor of $(m, w)$ is in $\mu'$ and that successor
is $(\mu'(w), w)$. Similarly, we have $\mu'(w) \geq_w m$ and thus we must
have $m >_w \mu(w)$. So only one successor of $(m, w)$ is in $\mu$: that
successor is $(m, \mu(m))$. This shows that the inequality~\ref{dumb_ineq} is tight for both $\mu$ and $\mu'$. 
Every successor $s$ of $(m, w)$ in $\Gamma$ that is not $(m, \mu(m)$ or $(\mu'(w),
w)$ is not in $\mu \cup \mu'$. So by $(a)$, for every such $s$, the vector $e_s$ is in
$E$. Similarly, since $(m, w) \notin \mu \cup \mu'$, we have $e_{(m, w)} \in E$. By subtracting the vector $e_s$ for every such $s$ and subtracting
the vector $e_{(m, w)}$ from the vector $\round{e_{(m,w)} + \sum_{v \text{ successor of } (m,w) \text{ in }\Gamma} e_v}$ given by the inequality \eqref{dumb_ineq}, we get the vector $(e_{(m, \mu(m))} + e_{(\mu'(w), w)})$, which must be in $E$.

If case (ii) of Claim 1 holds, the second part of the statement follows by a similar argument.\qed

\end{proof}

\begin{definition}
We say that $a \in \acceptable$ is \emph{linked} to $b \in \acceptable$ if $E$ contains a
vector of the form $(\lambda e_a + \gamma e_b)$, for some $\lambda, \gamma \in \RR$ such that $\lambda \neq 0$.
\end{definition}

The next lemma is immediate. 

\begin{lemma}
\label{link_reflexive} The following holds:
(i) For any $a \in \acceptable$, $a$ is linked to $a$; (ii) If $a$ is linked to $b$ and $b$ is linked to $c$, then $a$ is linked to $c$.
\end{lemma}

\begin{proof}
The first part follows by observing that $(e_a - e_a) = \vec{0} \in E$. 
For the second part, suppose that $p = (\lambda e_a + \gamma e_b)$ and $p' = (\lambda' e_b +
\gamma' e_c)$ are both in $E$. Then, we have
$
p - \frac{\gamma}{\lambda'} p' = \round{\lambda e_a - \frac{\gamma \gamma'}{\lambda'} e_c} \in E
$. \qed
\end{proof}

We now show how to use principal blocks of $\GStarMu$ to construct vectors
in $E$ that will serve our purpose.

\begin{lemma}
\label{compintersect}
If $P$ and $P'$ are two principal blocks with a vertex in common,
then for any $v \in V(P) \cap (\mu \cup \mu')$ and $u \in V(P') \cap (\mu \cup
\mu')$, $v$ is linked to $u$.
\end{lemma}

\begin{proof}
We distinguish two cases.

Case A: $P=P'$. If $v = u$, then $v$ is linked to $u$ by Lemma~\ref{link_reflexive}$(i)$.
Otherwise, let us assume that $P = P_{\mu, \mu'}^m$ for some man $m$.
Then, $v$ and $u$ both involve
$m$, and since $v$ and $u$ are both in $\mu \cup \mu'$, we must have
$\curly{v, u} = \curly{(m, \mu(m)), (m, \mu'(m))}$. Thus, $v$ is linked to $u$
by Lemma~\ref{lem:obs}(b).
The case $P = P_{\mu, \mu'}^w$ for some woman $w$ follows similarly, relying on
Lemma~\ref{lem:obs}(c).

Case B: $P\neq P'$. If there is some $(m, w) \in V(P) \cap V(P')$ such that $(m, w) \in \mu \cup \mu'$ then, using the arguments of Case A, 
both $v$ and $u$ are linked to $(m, w)$, hence $v$ is linked to $u$ by Lemma~\ref{link_reflexive}$(ii)$.
Assume now that the above does not hold, and let $(m, w) \in V(P) \cap V(P')$ with
$(m, w) \notin \mu \cup \mu'$. Without loss of generality 
let $P = P_{\mu, \mu'}^m$ and $P' = P_{\mu, \mu'}^w$.
By Lemma~\ref{lem:obs}(d), either $(e_{(m, \mu(m))} + e_{(\mu'(w), w)})$ or $(e_{(m, \mu'(m))} + e_{(\mu(w), w)})$ is in $E$.
Assume the first condition holds (the other case is similar). Then,
$(m, \mu(m))$ and $(\mu'(w), w)$ are linked to each other.
We have $v, (m, \mu(m)) \in V(P)$, so by Case A, $v$ is linked to
$(m, \mu(m))$. We have $(\mu'(w), w),
u \in V(P')$, so by Case A, $(\mu'(w), w)$ is linked to $u$.
Finally, using Lemma~\ref{link_reflexive}$(ii)$, we see that $v$ is linked to $u$.
\qed

\end{proof}

\begin{lemma}
\label{lem:final}
Let $u,v \in \mu \cup \mu'$ be two distinct vertices in the nontrivial component of
$\GStarMu$. Then $v$ is linked to $u$.
\end{lemma}

\begin{proof}
Consider a path in $\GStarMu$ between $v$ and $u$, with edges $e_1, \dots, e_k$ ($k\geq 1$).
 By Proposition~\ref{edge_to_circuit}, we can choose a principal block
$P_i$ such that $e_i \in P_i$. If $k=1$, then $u$ and $v$ are in  $V(P_1) \cap (\mu \cup \mu')$.
Using the arguments of Lemma~\ref{compintersect} (Case A), we can conclude that $v$ is linked to $u$. 

Assume now $k>1$. Note that
we must have $V(P_i) \cap (\mu \cup \mu') \neq \emptyset$, for all $1 \leq i \leq k$. To see this, assume without loss of generality
 that $P_i = P_{\mu, \mu'}^m$ for a man $m$. Since we have $\mu(m) \neq \mu'(m)$, at least
 one between $\mu(m)$ and $\mu'(m)$ is $\neq \bot$. Hence at least one between $(m, \mu(m))$ and $(m, \mu'(m))$ is in $V(P_i) \cap (\mu \cup \mu')$. 
For all $1 \leq i\leq k$, let $\ell_i\in V(P_i) \cap (\mu \cup \mu')$, with $\ell_1 = v$ and $\ell_k =u$.
Note that $e_i$ and $e_{i+1}$ share a vertex, and therefore $P_i$ and $P_{i+1}$ intersect (possibly, $P_i = P_{i+1}$).
By Lemma~\ref{compintersect}, $\ell_i$ is linked to $\ell_{i+1}$. Applying 
Lemma~\ref{link_reflexive}, we can conclude that $\ell_1$ is linked to $\ell_k$, i.e., $v$ is linked to $u$.
\qed
\end{proof}

Finally, we put things together and get the following:

\begin{corollary}
\label{boringcorollary}
Fix arbitrarily $u \in \mu \cup \mu'$ to be a vertex in the nontrivial component of
$\GStarMu$. Then every vertex $v$ of $\Gamma$ is linked to $u$.
\end{corollary}

\begin{proof}
If $v\notin \mu \cup \mu'$, then $v$ is linked to $u$ by Lemma~\ref{lem:obs}(a). \\
If $v \in \mu \cup \mu'$ and $v$ is a trivial component of $\GStarMu$, then $v \in \mu
\cap \mu'$ by Lemma~\ref{trivial_component_GStarMu}. Thus, the inequality $(-e_v + 0e_u, -1)$ is in $E'$. Therefore, 
the vector $(-e_v + 0e_u)$ is in $E$ and thus $v$ is linked to $u$.\\
Finally, if $v \in \mu \cup \mu'$ and $v$ is part of the nontrivial component of $\GStarMu$, 
$v$ is linked to $u$ by Lemma~\ref{lem:final}.\qed
\end{proof}

It is now easy to construct a set $S \subseteq E$ of linearly independent vectors, as follows.
Fix arbitrarily $u \in \mu \cup \mu'$ to be a vertex in the nontrivial component of
$\GStarMu$. Corollary~\ref{boringcorollary} gives that
for each vertex $v \in \Gamma$, $E$ contains a vector of the form $b_v = (\lambda_v e_v
+ \gamma_v e_u)$ with $\lambda_v \neq 0$.
Set $S:=\{b_v
\setb v \in \acceptable, v \neq u\}$. Then $|S|=|\acceptable|-1$, and the vectors in $S$ are clearly linearly independent.
This concludes the proof of the ``if'' part of Theorem~\ref{thetheorem}.
  
\subsection{Proof of the ``only if'' part}
We denote by $I(\mu, \mu')$ the set indexing all
the nontrivial components of $\GStarMu$. For any $i \in I(\mu, \mu')$, we let
$\GStarR$ be the corresponding component. Furthermore, for any matching $\mu'' \subseteq V(\GStarMu)$, let 
$\mu''_{|_i}$ be the restriction
to $\GStarR$.
The following lemma explains the key property of $\GStarMu$.

\begin{lemma}
\label{coollemma}
Let $\dot\mu$ be a matching that is a subset of $V(\GStarMu)$. Then $\dot\mu$ is stable if
\begin{itemize}
\item[(i)]
$\dot\mu$ contains all the isolated vertices of $\GStarMu$;
\item[(ii)]
for any $i \in I(\mu, \mu')$, $\dot\mu$ agrees with $\mu$ or $\mu'$ on the
component $\GStarR$, that is, $\dot\mu_{|_i}$ is either equal to
$\mu_{|_i}$ or to $\mu'_{|_i}$.
\end{itemize}
\end{lemma}

\begin{proof}
Let $(m, w) \in \acceptable$ be any acceptable pair. Since $\mu$ is stable, we have
$\mu(m) \geq_m w$ or $\mu(w) \geq_w m$. Since $\mu'$ is stable, we have $\mu'(m)
\geq_m w$ or $\mu'(w) \geq_w m$. Therefore, we have four possible cases, 
that without loss of generality reduce to the following two (up to symmetry).

\emph{Case 1: $\mu(m) \geq_m w$ and $\mu'(m) \geq_m w$.}
Clearly, ($\mu(m),m$) and ($\mu'(m),m$) are in the same component of
$\GStarMu$. If this component is an isolated vertex, we must have $\dot\mu(m) =
\mu(m) = \mu'(m) \geq_m w$ by the condition $(i)$.
Otherwise, this component is nontrivial, so by the second condition $(ii)$, we either
have $\dot\mu(m) = \mu(m)$ or $\dot\mu(m) = \mu'(m)$. In either case, we have
$\dot\mu(m) \geq_m w$.

\emph{Case 2: $\mu(m) \geq_m w$ and $\mu'(w) \geq_w m$.}
If $\mu(w) \geq_w m$ or $\mu'(m) \geq_m w$, we can reduce to the previous case, so
assume that $m >_w \mu(w)$ and $w >_m \mu'(m)$.
Then, we have $(m, w) \in V(\GStarMu)$.
Since $(m, w)$ is in $V(\GStarMu)$ and is adjacent to both $(m, \mu(m))$
and $(\mu'(w), w)$, those two vertices must be in the same component of
$\GStarMu$. By the second condition $(ii)$, we know that $\dot\mu$ either agrees with
$\mu$ or with $\mu'$ on this component, so we either have $(m, \mu(m)) \in
\dot\mu$ (and thus $\dot\mu(m) = \mu(m) \geq_m w$) or we have $(\mu'(w), w) \in
\dot\mu$ (and thus $\dot\mu(w) = \mu'(w) \geq_w m$).\qed

\end{proof}

The next corollary yields a proof of the ``only if'' part of Theorem~\ref{thetheorem}.

\begin{corollary}
\label{cor:only_if}
If $\GStarMu$ contains at least two nontrivial components, then $\mu$ and $\mu'$ correspond to two non adjacent extreme points of $\mathcal P_{SMT}$.
\end{corollary}

\begin{proof}
Let $I(\mu, \mu') :=\{1, \dots, k\}$, with $k \geq 2$. Fix an index $j$ in $I(\mu, \mu')$. Define
$\bar \mu$ as follows. For every isolated vertex $v$ of $\GStarMu$ we have $v \in
\bar\mu$. Furthermore, for $i=1,\dots,k$, we have:

\[
\bar\mu{}_{|_i} = \begin{cases}
\mu'_{|_i}, & \textrm{if } i = j \\
\mu_{|_i}, & \textrm{otherwise } 
\end{cases}
\]

Similarly, we define
$\bar \mu'$ as follows. For every isolated vertex $v$ of $\GStarMu$ we have $v \in
\bar\mu'$. Furthermore, we have:

\[
\bar\mu'{}_{|_i} = \begin{cases}
\mu_{|_i}, & \textrm{if } i = j \\
\mu'_{|_i}, & \textrm{otherwise } 
\end{cases}
\]
By Lemma~\ref{coollemma}, both $\bar \mu$ and $\bar \mu'$ are stable matchings, and hence correspond to
extreme points of $\mathcal P_{SMT}$.  Furthermore, $\bar \mu, \mu, \bar \mu', \mu'$ are all distinct, since $k>1$. 
One observes that 
$\frac{1}{2}x_{\mu} + \frac{1}{2}x_{\mu'} = \frac{1}{2}x_{\bar \mu} + \frac{1}{2}x_{\bar \mu'}$. This implies that
$x_{\mu}$ and $x_{\mu'}$ are not adjacent extreme points, since otherwise there would be a unique way to express their midpoint
as a convex combination of extreme points of $\mathcal P_{SMT}$.
\qed
\end{proof}

\section{Bounding the diameter}

Lemma~\ref{coollemma} clearly gives a strategy
 to get a path from $\mu$ to
$\mu'$ on the 1-skeleton of $\mathcal P_{SMT}$. Similarly to~\cite{stable_diameter}, 
we can change the coordinates of the corresponding points using one nontrivial component of
$\GStarMu$ at the time.

Let $I(\mu, \mu') = \curly{1,
\dots, k}$. Define a sequence of matchings $\mu^0,
\mu^1, \dots, \mu^k$ as follows. For any $0 \leq j \leq k$, choose
$\mu^j$ such that every isolated vertex $v$ of $\GStarMu$ is in 
$\mu^j$, and such that for every $i \in I(\mu, \mu')$, we have
\[
\mu^j{}_{|_i} = \begin{cases}
\mu'_{|_i}, & \textrm{if } i \leq j \\
\mu_{|_i}, & \textrm{if } i > j
\end{cases}
\]

Using Lemma~\ref{coollemma}, we can see that each $\mu^j$ is
a stable matching. Furthermore, note that $\mu^0 = \mu$ and $\mu^k = \mu'$. 
To show that $\mu^0, \dots, \mu^k$ is a path of length $k$ from $\mu$
to $\mu'$, it remains to show that for every $j < k$, $\mu^j$ and
$\mu^{j+1}$ are adjacent. This is done in the next lemma.

\begin{lemma}
\label{lem:adjacent}
For every $j < k$, $\mu^j$ and
$\mu^{j+1}$ are adjacent.
\end{lemma}

\begin{proof}
First, observe that 
every vertex of $\GStarRjj$ is in $\GStarDotMu$, since
$\mu^j{}_{|_{j+1}} = \mu_{|_{j+1}}$ and
$\mu^{j+1}{}_{|_{j+1}} = \mu'_{|_{j+1}}$.
Second, observe that every vertex of $\GStarDotMu$ is in
$\mu^j \cap \mu^{j+1}$ or adjacent to a vertex in $\GStarRjj$.
To show this, let $v$ be some vertex in $V(\GStarDotMu)$ such that
$v \notin \mu^j \cap \mu^{j+1}$.
By Lemma~\ref{trivial_component_GStarMu}, $v$ is not an isolated vertex of
$\GStarDotMu$, so it has to have a neighbor $v' \in V(\GStarDotMu)$.
Without loss of generality, assume the edge between $v$ and $v'$ represents the preference of a man
$m$. Let $v = (m, w)$ and $v' = (m, w')$ (where $w \neq w'$). 
Note that this implies that
$\mu^j(m) \neq \mu^{j+1}(m)$.
Thus, we must have $\mu^j(m) = \mu(m)$ and $\mu^{j+1}(m) = \mu'(m)$.
Thus, either $(m, \mu(m))$ or $(m, \mu'(m))$ is a vertex of $\GStarRjj$, and $v$ is adjacent to this
vertex.

From the above two observations, we can conclude that $\GStarDotMu$ has exactly one nontrivial
component: from the first one we can conclude that there is one nontrivial
component $\Phi$ (the one containing $\GStarRjj$). From the second one, we can conclude
that every vertex of $\GStarDotMu$ is either adjacent to some vertex in $\GStarRjj$
(and thus in $\Phi$) or it is in $\mu^j \cap \mu^{j+1}$ (and thus it is an isolated
vertex of $\GStarDotMu$).
Theorem \ref{thetheorem} lets us conclude that $\mu^j$ and
$\mu^{j+1}$ are adjacent. \qed

\end{proof}

The above discussion and Lemma~\ref{lem:adjacent} yield a proof of the following theorem.

\begin{theorem}
Let $\mu$ and $\mu'$ be two arbitrary stable matchings. The distance between 
$\mu$ and $\mu'$ on the 1-skeleton of $\mathcal P_{SMT}$ is at most $\abs{I(\mu,
\mu')}$.
\end{theorem}

As a corollary, we get the following bound on the diameter of $\mathcal P_{SMT}$.
\begin{corollary}
\label{cor:final}
Let $n := |M \cup W|$. The diameter of $\mathcal P_{SMT}$ is at most $\floor{\frac{n}{3}}$.
\end{corollary}

\begin{proof}
The smallest nontrivial component of $\GStarMu$ (for any $\mu, \mu'$)
has two vertices, and thus involves three people (it could represent either a man indifferent
between two women, or a woman indifferent between two men). Thus the number of
nontrivial components is at most $\frac n 3$.\qed
\end{proof}

If we do not have ties, then the smallest nontrivial component of $\GStarMu$ (for any $\mu, \mu'$)
involves at least two men and two women (see~\cite{stable_diameter}). Thus, in this case the bound for the diameter can be strengthened to $\frac n 4$ (in fact, to $\min \left\{ \frac{|M|}{2},\frac{|W|}{2} \right\}$), as shown in~\cite{stable_diameter}.

Eventually, note that it is easy to construct a family of instances for which the bound of Corollary~\ref{cor:final} is tight.
Consider a set of $t$ men $M:=\{m_1, \dots, m_t\}$ and 
$2t$ women $W:=\{w_1, \dots, w_{2t}\}$. 
Let $P(m_i)=\{w_i, w_{i+t}\}$ with $w_i \indiff_{m_i} w_{i+t}$, for $i=1,\dots,t$.
Let $P(w_i)= P(w_{i+t})=\{m_i\}$, for $i=1,\dots,t$.
Let $\mu$ be the stable matching given by all pairs of the form $(m_i,w_i)$ for $i=1,\dots,t$, and 
$\mu'$ be the stable matching given by all pairs of the form $(m_i,w_{i+t})$ for $i=1,\dots,t$. One checks
that $\mu$ and $\mu'$ are stable matchings, and the corresponding extreme points are at distance $t$ on $\mathcal P_{SMT}$. 

\bibliographystyle{splncs04}
\bibliography{references}
\end{document}